\documentclass[a4paper,preprint,12pt,authoryear]{elsarticle}
\pdfoutput=1

\usepackage{hyperref}
  \usepackage{natbib}
\usepackage[svgnames]{xcolor}
\usepackage{amsmath}
\usepackage{amsthm}
\usepackage{amssymb}
\usepackage{pdfpages}
\usepackage{enumitem}

\usepackage{hyphenat}
\usepackage{algorithm}
\usepackage{algpseudocode}
\usepackage{algorithmicx}\usepackage{dutchcal}

\numberwithin{equation}{section}

\newcommand{\mat}[1]{{#1}} 
\newcommand{\pmat}[1]{{\mathcal{#1}}}

\newcommand{\companion}{{\phi}}
\newcommand{\DCons}{\Gamma}

\newcommand{\N}{\mathcal{N}}

\newcommand{\A}{{\pmat{A}}}
\newcommand{\B}{{\pmat{B}}}
\renewcommand{\b}{{{\pmat{b}}}}
\newcommand{\bb}{{{\pmat{b}}}}

\newcommand{\blin}{{\widehat{\mathcal{b}}}}
\newcommand{\Alin}{{\widehat{\mathcal{A}}}}
\newcommand{\DeltaAlin}{{\widehat{\Delta\mathcal{A}}}}
\newcommand{\Blin}{{\widehat{\mathcal{B}}}}

\newcommand{\Span}{\mbox{\upshape{span}}}

\DeclareMathOperator{\rank}{rank}
\DeclareMathOperator{\diag}{diag}

\renewcommand{\vec}{{\mbox{\upshape vec}}}

\newcommand{\RR}{{\mathbb{R}}}
\newcommand{\CC}{{\mathbb{C}}}
\newcommand{\QQ}{{\mathbb{Q}}}
\newcommand{\FF}{{\mathbb{F}}}

\newcommand{\D}{\partial}
\newcommand{\norm}[1]{{\|#1\|}}
\newcommand{\tallnorm}[1]{ {\left \|#1 \right \|}}
\newcommand{\nxn}{{n\times n}}

\definecolor{darkgreen}{rgb}{0,.35,0}
\definecolor{darkblue}{rgb}{0,0,.5}
\definecolor{darkred}{rgb}{.6,0,0}

\newcommand\scalemath[2]{\scalebox{#1}{\mbox{\ensuremath{\displaystyle #2}}}}

\newtheorem{theorem}{Theorem}[section]
\newtheorem{corollary}[theorem]{Corollary}
\newtheorem{lemma}[theorem]{Lemma}

\newtheorem{definition}[theorem]{Definition}
\newtheorem{fact}[theorem]{Fact}
\newtheorem{example}[theorem]{Example}
\newtheorem{problem}[theorem]{Problem}
\newtheorem*{problem*}{Problem}

\newtheorem{conjecture*}{Conjecture}

\newtheorem*{remark*}{Remark}

\date{\today}

\begin{document}

%

\begin{frontmatter}

\title{Computing Lower Rank Approximations of Matrix Polynomials}

\author[uw]{Mark Giesbrecht}
  \ead{mwg@uwaterloo.ca}
\author[uw]{Joseph Haraldson}
  \ead{jharalds@uwaterloo.ca}
    \author[uw]{George Labahn}
  \ead{glabahn@uwaterloo.ca}
\tnotetext[uw]{Cheriton School of Computer Science,
    University of Waterloo, Waterloo, Ontario, Canada}


\begin{abstract}
Given an input matrix polynomial whose coefficients are floating point numbers, we consider the 
problem of finding the nearest matrix polynomial which has rank at most a 
specified value. This generalizes the problem of finding a nearest 
matrix polynomial that is algebraically singular with a prescribed lower bound 
on the dimension given in a previous paper by the authors. In this paper we prove that such 
lower rank matrices at minimal distance always exist, 
satisfy regularity conditions, and  are all isolated and surrounded by a basin of 
attraction of non-minimal solutions.  In addition, we present an
iterative algorithm which, on given input sufficiently close to a rank-at-most 
matrix, produces that matrix.  The algorithm is efficient and is proven to 
converge quadratically given a sufficiently good starting point.  An 
implementation demonstrates
the effectiveness and numerical robustness of our algorithm in practice. 
\end{abstract}
\end{frontmatter}

\section{Introduction}

Matrix polynomials appear in many areas of computational algebra, control 
systems theory, differential equations, and mechanics. The algebra of matrix 
polynomials is typically described assuming that the individual polynomial 
coefficients come from an exact arithmetic domain. However, in the case of 
applications these coefficients typically have numeric coefficients, usually 
real or complex numbers. As such, arithmetic can have numerical errors and 
algorithms are prone to numerical instability.  

Numerical errors have an impact, for example, in determining the rank of a matrix polynomial
with floating point coefficients. In 
an exact setting determining the rank or determinant of a matrix polynomial is 
straightforward, and 
efficient procedures are available, for example from \cite{StoVil05}. However, 
in a numeric environment,  a matrix polynomial may appear to have full or high 
rank while at the same 
time being  close to one having lower rank.  Here ``close'' is defined 
naturally 
under the Frobenius norm on the underlying coefficient matrices of the matrix  
polynomial. Rather than 
computing the rank of the given matrix polynomial exactly, one can ask how far 
away it is from one that is rank-deficient, and then  to find one at that 
distance. In the case of 
matrices with constant entries this  is a problem solved via the Singular Value 
Decomposition (SVD).  However, in the case of matrix polynomials no equivalent 
rank revealing factorization has thus far been available.

In this paper we consider the problem of computing the nearest matrix polynomial 
to an input matrix polynomial in $\RR[t]^{m\times n}$ having rank at most a 
specified value $r$.  More precisely, given an integer $r$ and an  $\pmat{A}\in 
\RR[t]^{m\times n}$ of full rank, we want to compute $\Delta \pmat{A}\in 
\RR[t]^{m\times 
n}$ with $\deg(\Delta\A_{ij})\leq \deg \A_{ij}$ (or similar degree constraints 
to be specified later), such that  $\pmat{A}+\Delta \pmat{A}$  has rank at most 
$n-r$ and where   $\norm{\Delta\pmat{A}}$ is minimized.   In the case where 
$n-r$ is one less than 
the row or column size then this is the problem of finding the nearest matrix 
polynomial which is \emph{singular}. 

A reasonable metric for measuring closeness on the space of matrix polynomials 
over the reals is the Frobenius norm.  For a
matrix polynomial
$\A\in\RR[t]^{m\times n}$, with $(i,j)$ entry $A_{ij}\in\RR[t]$, the
\emph{Frobenius} norm is given by 
\begin{equation}\label{frobNorm}
  \norm{\A}^2 = \norm{\A}^2_F = \sum_{1\leq i\leq m,1\leq j \leq n} 
\norm{A_{ij}}^2,
\end{equation}
where, for a polynomial $a\in\RR[t]$, the coefficient 2-norm is defined by
\begin{equation}\label{polyNorm}
  a=\sum_{0\leq i\leq \deg a} a_i t^i,
  ~~~~~~~~\norm{a}^2 = \norm{a}^2_2 = \sum_{0\leq i\leq \deg a} a_i^2.
\end{equation}

The main results in this paper center on the characterization of the geometry of 
minimal solutions.  We show that minimal solutions exist, that is, for a given 
$r$ there exists a 
$\Delta\A\in\RR[t]^{m\times n}$ of minimal norm such that $\A+\Delta\A$ has rank 
at most $n-r$
and meets the required degree  constraints on perturbed coefficients. In 
addition, we show that minimal solutions are isolated and are surrounded by a 
non-trivial open neighbourhood of non-minimal solutions. Also regularity and 
second-order 
sufficiency conditions are generically satisfied and a restricted version of  
the problem always satisfies these conditions. Finally we show that we can also 
generalize our results to the lower rank approximation instance of matrix 
polynomials generated by an affine  structure\footnote{A matrix  $A \in 
\FF^{m\times n}$ has an affine structure over a ring $\FF$ if it can be written 
as  $  A = B_0 + \sum_{i=1}^L c_i B_i$ for $\{B_0,B_1,\ldots, B_L \} \subseteq 
\FF^{m\times n}$ and $c_i \in \FF$. If $B_0$ is the zero matrix, then the 
structure is said to be linear. Examples of linear structures include symmetric 
and  hermitian matrices while   matrices with an affine structure include 
entries that are fixed non-zero coefficients, such as monic matrix 
polynomials.}, and so generalize to low-rank 
approximations of structured matrices by taking the degree to be zero. 

We demonstrate efficient algorithms for computing our minimal lower rank approximants. 
That is, 
 for an input matrix polynomial 
$\A\in\RR[t]^{m\times n}$ (with prescribed affine structure) sufficiently close 
to a singular matrix polynomial, we give an 
iterative scheme which converges to a rank at most matrix polynomial at minimal 
distance, at a provably quadratic rate of 
convergence.  We further generalize the iterative scheme so that it converges to 
a matrix polynomial with a kernel of dimension at least $r$, at a  minimal 
distance and a provable quadratic rate of convergence.  Finally, we also discuss 
a 
Maple implementation which demonstrates the convergence and numerical robustness 
of our iterative scheme.
  
\subsection{Previous research}

Much of the work in this area has often been done under the heading of 
\emph{matrix pencils}. 
See \cite{GolLanRod09} for an 
excellent overview.  Non-singular (full rank) square matrix polynomials
are sometimes referred to as \emph{regular matrix polynomials}.

In the case of finding  the nearest  singular matrix pencil this problem was 
solved by the present authors in \cite{GHL17}. Previous to that this problem was 
posed for linear matrix pencils in \cite{ByeNic93} and followed up in 
\cite{ByeHeMeh98}.  The nearest 
singular matrix polynomial relates to the stability of polynomial eigenvalue 
problems, linear time invariant systems and differential-algebraic equations 
studied subsequently in 
\citep{KreVoi15,GulLubMeh16}.  For non-linear matrix polynomials/pencils, 
previous works rely on embedding a non-linear
(degree greater than 1) matrix polynomial into a linear matrix polynomial of 
much higher order.  Theorem 1.1 and Section 7.2 of
\cite{GolLanRod09} shows that any regular $\A\in\RR[t]^\nxn$ of degree $d$, is 
equivalent to a linear matrix polynomial 
$\B=B_0+tB_1$, for $B_0,B_1\in\RR^{nd\times nd}$.   However, this equivalence is 
(obviously) not an isomorphism, nor is it distance preserving\footnote{The 
equivalence mapping is not surjective.}.  Hence a nearby singular matrix 
polynomial to 
$\B\in\RR[t]^{nd\times nd}$ (even when constrained to a degree one perturbation) 
almost certainly does not correspond to a
nearby singular matrix polynomial to $\A\in\RR[t]^\nxn$.  Moreover, even if one 
was to perturb to a
rank-reduced matrix within the image of the linearization, the inverse image 
would not necessarily have reduced rank.  
In \cite{LawCor15} a more sophisticated linearization with an eye towards 
ameliorating this is explored.

In the context of computer algebra the notion of symbolic-numeric
algorithms for polynomials has been an active area of research for a
number of years, and the general framework of finding nearby instances
with a desired algebraic property is being thoroughly explored.
Closest to our work here is work on approximate Greatest Common Divisors (GCD)
\cite{CGTW95,BecLab98,BecLab98b}, multivariate polynomial factorizations 
\cite{KalMayYangZhi08}, and especially the optimization-based approaches
employing the Structured Total Least Norm algorithm 
\cite{LiYan05,KalYan05,KalYan06,Zhi07} and Riemannian SVD
\cite{BotGie05}. 
More recently, we have explored computing the approximate
GCRD of (non-commutative) differential polynomials
\citep{GieHar15,GieHar16} and resolve similar issues.

The computer algebra community has made impressive progress on fast,
exact algorithms for matrix polynomials, including nearly optimal
algorithms for computing ranks, factorizations and various normal
forms; see \cite{KalSto15} and references therein for a recent
overview.  Part of our goal in this current paper is establish a basis
for extending the reach of these symbolic techniques to matrices of
polynomials with floating point coefficients.

In a more general setting our problem can be formulated as a
Structured Low Rank Approximation (SLRA) problem.  A popular method to
solve SLRA problems is the Structured Total Least Norm (STLN) approach
\citep{RosParGli96, RosGliPar98}.  These are iterative methods and in
general their convergence to stationary points is linear (first
order), rather than quadratic, unless additional assumptions are
made. In the event STLN converges to a solution, there may be other
solutions arbitrarily nearby, as second order sufficient conditions may not 
hold.
The SLRA problem is a non-linear least squares problem and accordingly other
techniques such as the Restricted and Riemannian SVD
\citep{DeM93,DeM94,DeM95} provide general tools for solving such
problems.
Other heuristic tools applicable to our problem include 
variable projection \citep{GolPer73,GolPer03} and Newton's method
\citep{AbaMenHar91}.  We would expect these methods to perform very poorly in
our case, as one can expect problems with large residuals to perform
poorly and the rational function arising from variable projection can
be too costly to deal with for modestly sized problems.  The problem
may also be considered as optimization on a manifold
\citep{AbsMahSep09}, however we do not explicitly consider this
approach. For a detailed survey of affinely structured low-rank
approximation, see \citep{Mar08,Mar11}.

Other methods for structured low-rank approximation involve the family
of lift and project algorithms, with the best known being Cadzow's
algorithm \citep{Cad88}.  More recently \cite{SchSpa16} gives a
sequence of alternating projections that provably converge
quadratically to a fixed point. However, lift and project algorithms
do not generally satisfy necessary first order (see \citep{Ber99}) optimality 
conditions,
and while they may converge (quickly) to a fixed point, there is no
guarantee that the fixed point is an optimal solution, though it is
usually quite good.  In any case, for specific problems such as ours,
understanding the geometry of the minimal solutions (and hence the
well-posedness of the problem) is key to effective algorithms for
their computation. 


SLRA problems are in general NP-hard to solve, see for example \citep{PolRoh93,BraYouDoyMor94}. 
They are also hard to  approximate under affinely structured matrices over 
$\QQ$. 
In general the hardness stems from determining 
if a bilinear system of equations admits a non-trivial solution. In the 
instance of classical matrix polynomials it is trivial to construct feasible 
points since the underlying scalar matrix problem is linearly structured. 

All of our contributions apply to matrix polynomials with an  affine structure 
provided that feasible points exist, 
that is, singular matrix polynomials with a prescribed structure exist, which is 
NP-hard in general. In particular, in 
the degree zero case our algorithms and techniques apply to affine SLRA 
problems. Thus, computing the nearest (affinely structured) 
matrix polynomial is equivalent to SLRA problems with an affine structure.   


While the  contributions in this paper focus on local properties of SLRA, the 
local properties also imply global results.  
The Sum of Squares (SOS) hierarchy is a global framework for studying polynomial 
optimization problems  subject to polynomial
constraints \cite{Las01}. The SOS optimization tools have found experimental 
success in computing structured distances to
singularity and extracting minimizers when the solutions are locally unique, see 
for example  \cite{HenLas06}. In general the SOS hierarchy
converges for an infinite order of relaxations, but for several problems the 
relaxations converge after a finite order. The finite convergence is in 
polynomial time with respect to the input and the number of relaxations.  In 
particular, this
finite convergence was observed for affine SLRA problems in \cite{HenLas06} but 
little theory was provided to indicate the
reason why. The later work of \cite{Nie14} shows that, under regularity and 
second-order sufficiency conditions, finite
convergence always occurs and that it is possible to extract a minimal solution. 
In our contributions we prove that 
second-order sufficiency and regularity conditions hold generically (and if they 
do not, then they will hold on a restricted
subset of the problem). The corollary to this is that the SOS hierarchy will 
have finite convergence for  affine SLRA problems if a solution exists (such as 
computing the distance of the nearest rank-deficient matrix polynomial) and if 
the embedding is 
minimal then a minimizer may be extracted as well. Another useful feature of 
the SOS hierarchy is even if convergence cannot be certified, a structured 
lower-bound is obtained.
%


\subsection{Outline}
%

In Sections \ref{sec:geometry} and \ref{sec:rankFact} we describe 
tools needed for our constructions and then explore the geometry of
our problem.  We show that the problem is locally well-posed. One
cannot expect the nearest rank at most matrix polynomial to be
unique. However under weak normalization assumptions, we show that
solutions are locally unique in a closed-ball around them.  To
complement the separation of solutions, we also show that for an equivalent 
problem, solutions
corresponding to a different closed ball are separated by at least a
constant amount independent of the dimension of the space.

In Section \ref{sec:algorithm} we give an equality constrained variant
of Newtons' method for computing via post-refinement the nearest
rank at most matrix polynomial.  The main idea is to
compute an initial guess with a suitable first order or lift-and
project method.  We are able to prove that, with a suitable initial
guess and regularity assumptions, our algorithm generally has local
quadratic convergence except for degenerate cases.  This is done by
deriving closed-form expressions for
the Jacobian of the constraints and the Hessian of the Lagrangian. When we 
refer to the 
speed of convergence, we refer to quotient rates as is typical in the 
nomenclature.
 
In Section \ref{sec:implementation} we describe our prototype
implementation, including heuristics for starting points and other
improvements. We discuss the numerical performance of the algorithm
and give examples demonstrating convergence.  
results for a low-rank approximation of matrix polynomials. 
The paper ends  with a conclusion and topics for future research.

%


\section{Preliminaries and Geometry}
\label{sec:geometry}


In this section we will introduce some basic definitions and explore
the numerical geometry of our lower rank problem.  Canonically we will let
\[
\A= \sum_{j=0}^d  A_j t^j \in \RR[t]^{n\times n}
\]
be a matrix polynomial, with coefficients $A_0,\ldots,A_d\in\RR^\nxn$.
In the case of rectangular matrix polynomials 
we are able to pad the matrix with zeros, thus embedding the problem into one with square matrix polynomials.
Thus we will let
\[
\A= \sum_{j=0}^d  A_j t^j \in \RR[t]^{n\times n}
\]
be a matrix polynomial, with coefficients $A_0,\ldots,A_d\in\RR^\nxn$.
The \emph{degree} $\deg\A$ of $\A$ is defined as $d$, assuming that $ \pmat{ A}_d \neq 0$.
      
 We say that $\A$ is singular if $\det(\A)$ is the zero polynomial in $\RR[t]$, or equivalently, that there is a
$\pmat{b} \in \RR[t]^{n\times 1}$ such that $\A\pmat{b}\equiv 0$.  The kernel of $\A$ is
$\ker \A = \{ \pmat{b}\in \RR[t]^{n\times 1} : \A \pmat{b} = 0
\}$
and the rank of $\A$ is $n-\dim \ker \A$ (as a vector space over $\RR(t))$. Then 
$\A$ has rank at most $n-r$ if there
exists at least $r$ linear independent vectors $\{\pmat{b}_i\}_{i=1, ... , 
r}$ satisfying $\A \pmat{b}_i = 0$.

For $a \in \RR[t]$, define
\begin{equation}
\label{eq:companion}
\companion(a) = \companion^{(n,d)}(a) =
\begin{pmatrix}
  a_0  \\
  a_1 & a_0 \\
  \vdots &  & \ddots \\
  a_d    &  &  & a_0 \\
  & a_d & & a_1 \\
  & &  \ddots &   \vdots \\
  & & & a_d
\end{pmatrix} \in \RR^{(\mu+d) \times \mu},
\end{equation}
where $\mu=nd+1$.  $\companion(a)$ is a Toeplitz matrix. Such matrices are conveniently used to describe polynomial multiplication in the sense that if $c = a \cdot b$ 
with $a$ of degree $d$ and $c\in\RR[t]$ of degree at most $\mu-1$, then $\vec(c) = \companion(a) \cdot \vec(b)$ where $\vec(p)$ is the vector of coefficients of a polynomial.

\begin{definition}
  The \emph{$\RR$-embedding} of $\A\in \RR[t]^{n\times n}$ is
  \[  
  \Alin = \begin{pmatrix}
    \companion (A_{1,1}) & \cdots & \companion(A_{1,n}) \\
    \vdots & &  \vdots \\
    \companion (A_{n,1}) & \cdots & \companion(A_{n,n})
  \end{pmatrix} \in \RR^{ n(\mu+d) \times n \mu}.
  \]
  For $\pmat{b}\in\RR[t]^{n\times 1}$ of degree $\mu-1$ the
  $\RR$-embedding of $\pmat{b}$ is
  \[ 
  \blin = (b_{1,0}, b_{1,1}, \ldots, b_{1,\mu-1}, \ldots, b_{n,0},
  \ldots,b_{n,\mu-1})^T \in \RR^{n \mu\times 1}.
  \]
\end{definition}

Note that $\A \cdot b=0$, for $\pmat{b}\in\RR[t]$ of degree at most
  $\mu-1$ if and only if $\Alin \cdot \blin={0}\in\RR^{n\mu\times 1}$. 
  This property is central to our work in the coming sections.
  
  For ease of notation we will take 
\[N= n(\mu+d)=n^2d+n(d+1), ~~M= n \mu=n^2d+n \mbox{ and }  
R \geq 1\] when dealing with $\RR$-embeddings in subsequent sections. We
note that $\Alin$ is a block-Toeplitz matrix, and as such one method
of understanding the lower rank problem is to find close by structured rank
deficient block-Toeplitz matrices, a typical structured low rank
approximation problem. Some authors refer to such embeddings as a 
(permuted) Sylvester matrix associated with $\A$. We avoid this terminology as 
it is ambiguous when considering Sylvester matrices occurring in (approximate) 
GCD computations.

Unlike the standard linearizations in \citep[Section 7.2]{GolLanRod09}
used to turn arbitrary degree matrix pencils into linear pencils, this
$\RR$-embedding is kernel preserving for matrix polynomials of arbitrary
degree.  In particular, $\bb\in \ker \A$ with $\deg \bb \leq \mu$ implies 
$\blin\in\ker\Alin$.
The $\RR$-embedding is also quasi-distance preserving, since
$\norm{\A}_F^2 = \frac{\norm{\Alin}_F^2}{\mu}$.



\medskip


\begin{problem} 
  \label{prb:problem-original}
  \noindent\textbf{Main Problem:} \\
  Given $\A\in\RR[t]^\nxn$ non-singular of degree $d$  and an integer $r \leq n$
  determine $\Delta \A \in \RR[t]^{n\times n}$, with
  $\deg \Delta\A_{ij}\leq \deg \A_{ij}$ for all $1\leq i,j\leq n$,
  and $n~-~r$ linearly independent vectors $\pmat{b}_k \in\RR[t]^{n\times 1}$, such that $\norm{\Delta \A}$ is
  (locally) minimized, subject to the constraint that $(\A + \Delta \A)\pmat{b}_k = 0$ and $\|\pmat{b}_k  \| = 1$.
\end{problem}

Note that this is minimizing a convex objective function subject to non-convex 
constraints. However, the equality constraints are linear in each argument.
It is still not clear that Problem~\ref{prb:problem-original} is
well-posed in the current form. We will prove that solutions exist,
that is, there is an attainable global minimum value and not an
infimum.


\begin{lemma}\label{lem:degree-bounds}
  $\A\in \RR[t]^{n\times n}$ is singular if and only if there exists a
  $\pmat{b}\in \RR[t]^{n\times 1}$ with
  $\deg \pmat{b} \leq nd = \mu-1$ such that $\pmat{Ab}=0$.
\end{lemma}

\begin{proof}
  Suppose that $\A$ has rank $s<n$.  By permuting rows and columns we
  may assume without loss of generality that the leading $s\times s$
  submatrix of $\A$ is non-singular.  There is a unique vector of the
  form \[\pmat{c}=(b_1/\gamma,\ldots,b_s/\gamma,-1,0,\ldots,0)\] from
  Cramer's rule such that $\A\pmat{c}=0$, where $\gamma\in\RR[t]$ is
  the determinant of the leading $s\times s$ minor of $\A$, and all of
  $b_1,\ldots,b_s,\gamma\in\RR[t]$ have degree at most $sd\leq nd$.
  Multiplying through by $\gamma$, we find that
  $\pmat{b}=\gamma\pmat{c}$ satisfies the requirements of the lemma.
\end{proof}
See \cite[Corollary~5.5]{BecLabVil06} for an alternative proof.


\begin{lemma}
  \label{lem:rank-def}
  $\A$ is singular if and only if $\Alin$ does not have full column rank.
\end{lemma}
\begin{proof}
  If $\A$ is rank deficient then there exists
  $\pmat{b} \in \RR[t]^{n\times 1}$ with $\deg \pmat{b} \leq \mu-1$
  such that $\pmat{Ab}=0$. $\Alin$ has a non-trivial kernel and,
  $\blin \in \ker \Alin$ by construction. Conversely, suppose that
  $\A$ has full rank. Then for all $\pmat{b}\in \RR[t]^{n\times 1}$ we
  have $\pmat{Ab} \neq 0$ which implies that $\Alin \blin \neq 0$ or
  $\ker \Alin$ is trivial.
\end{proof}

We recall the Singular Value Decomposition as the
primary tool for finding the distance to the nearest
\emph{unstructured} rank deficient matrix over $\RR$ or~$\CC$.

\begin{definition}
  A Singular Value Decomposition (SVD) of
  $\mat{C} \in \RR^{N\times M}$ is given by
  $\mat{C} = \mat{Q} \cdot \Sigma \cdot \mat{P}^T$, where
  $\mat{Q} \in \RR^{M\times M},\mat{P}^T \in \RR^{N\times N}$ are
  orthogonal matrices and $\Sigma =$
  $\diag(\sigma_1,\ldots,\sigma_{M}) \in \RR^{M\times N}$ is a
  diagonal matrix consisting of the singular values of $\mat{C}$ in
  descending order of magnitude.  See \citep{GolLoa12}.
\end{definition}

The following fact is a standard motivation for the SVD.
\begin{fact}[\cite{EckYou36}]
  \label{fact:nearest}
  Suppose $C=Q\Sigma P^T\in\RR^{N\times M}$ as above has full column rank,
  with $N\geq M$.  Then $\Delta C=Q \diag(0,\ldots,0,-\sigma_{M}) P^T$
  is such that $C+\Delta C$ has column rank at most $M-1$,
  $\norm{\Delta C}_F=\sigma_M$, and $\Delta C$ is a perturbation of minimal
  Frobenius norm which reduces the column rank of $C$.
\end{fact}
\begin{lemma}
  \label{lem:nearest-unstructured-distance}
  Given a non-singular $\A\in\RR[t]^\nxn$, and 
  $\Delta\A\in\RR[t]^\nxn$ such that $\B=\A+\Delta\A$ is singular, it
  is the case that $\norm{\DeltaAlin}\geq \sigma_{n\mu}(\Alin)$.
\end{lemma}
\begin{proof}
  By Lemma \ref{lem:rank-def} above, $\Blin$ is not of full column
  rank.  Thus, by Fact \ref{fact:nearest}
  $\norm{\DeltaAlin}_F\geq \sigma_{n\mu}(\A)$.
\end{proof}


\begin{corollary}
  The set of  rank $r>1$ matrices over $\RR[t]^{n\times n}$ of degree
  at most $d$ is open, or equivalently, the set of all matrices of rank 
  at-most $n-r$ over $\RR[t]^{n\times n}$ of degree at most $d$ is closed.
\end{corollary}

\begin{theorem}[Existence  of Solutions]
  The minimization posed in Problem~\ref{prb:problem-original} has an
  attainable global minimum if
  $\deg \Delta \pmat{A}_{i,j} \leq \deg \pmat{A}_{i,j}$ for all $1\leq
  i,j\leq n$.
\end{theorem}
\begin{proof}
  Let 
  \begin{align*} 
    S= & \left \{ \pmat{C} \in \RR[t]^{n\times n} ~|~
         \rank \pmat{C} \leq
         n-r \wedge \deg \pmat{C} \leq  d\right\} \\
       &\cap \left\{ \pmat{C} \in \RR[t]^{n\times n}
         | \norm{\pmat{C}}_F^2\leq \norm{\A}_F^2
         \right\}.
  \end{align*}
  $S$ is the intersection of a closed and bounded set and a closed
  set, hence $S$ is closed and bounded. $S$ is isomorphic to some
  closed and bounded subset of Euclidean space, hence by the
  Heine-Borel theorem, $S$ is compact.  To show the set is non-empty,
  we note that, by the degree assumption on $\Delta \A$,
  $\Delta \A = -\A$ is a feasible point independent of rank.

  Let $\pmat{C} \in S$ then
  $\norm{\A-\pmat{C}}_F^2 = \norm{\Delta \A}_F^2$ is a continuous
  function over a compact set. By Weierstrass' theorem it has an
  attainable global minimum.
\end{proof}


It is important not to over-constrain the problem with a choice of
$\Delta \A$, since otherwise the feasible set might be empty. Another
reasonable choice of $\Delta \A$ which we can handle, is that the
perturbation has the same coefficient structure/support as $\A$, that is, zero
terms in polynomial entries are preserved.

We note that this result says nothing about uniqueness or separation of 
solutions or any local properties. All that has been shown is that if the 
perturbations are in the same space as the input, and one seeks a rank at-most 
approximation, then there is an attainable global minimum 
value, i.e. not an infimum. If one wants a minimal solution  with the rank 
being exactly $n-r$,  
then there is no guarantee that there is an  attainable global minimum to 
Problem~\ref{prb:problem-original}.

\section{Rank Factorization}\label{sec:rankFact}

A natural formulation of the problem that encompasses the rank
implicitly is to perform a rank factorization and write
$\pmat{A}+\Delta \pmat{A} = \pmat{UV}$ for
$\pmat{U}\in \RR[t]^{n\times (n-r)}$ and $\pmat{V}\in \RR[t]^{(n-r)\times
  m}$. Here $\pmat{UV}$ is subject to some constraints that preserve
the structure of $\Delta \pmat{A}$ (i.e., that we do not perturb any
coefficients we are not allowed to, typically that
$\deg\Delta\A_{ij}\leq \deg\A_{ij}$, but possibly also preserving the
zero coefficients and not introducing a larger support). This is a
non-linear least squares problem. 
However solutions are not unique. Indeed, if
$\pmat{Z}\in \RR[t]^{(n-r)\times (n-r)}$ is unimodular (i.e.,
$\det(\pmat{Z})\in \RR^*$), then $\pmat{UZ}$, $\pmat{Z}^{-1}\pmat{V}$
is another rank $n-r$ factorization, and we obtain an infinite family.
While normalizing over matrix polynomial rank-factorizations is
difficult, it is much easier to exploit the quasi-distance preserving
property of $\norm{\cdot}_F$ and look at rank-factorizations of
$\Alin$, that do not necessarily correspond to $\pmat{U}$ and
$\pmat{V}$.

%

\subsection{Embedded Rank Factorization}

\begin{definition}
  Let $N=(\mu+d)n$, $M=n\mu$ and $R>0$. A rank factorization
   of $\Alin+\DeltaAlin$ is given by writing $\Alin +\DeltaAlin=
   UV$ where  $U \in \RR^{N\times R}$ and $V \in \RR^{R\times M}$ are
   arbitrary (unstructured) matrices over $\RR$.
\end{definition}
Our goal is to find $U,V$ with appropriate dimensions which minimize 
\[
\norm{\Delta\Alin} = \norm{\Alin-UV}
\]
\emph{and} such that $\Delta\Alin$ has the correct Toeplitz-block
structure (i.e., it is an $\RR$-embedding of a matrix polynomial).
This is a problem with a non-convex objective function (that is convex
in each argument) and non-convex constraints. We note that $U$, $V$
have no direct connection with
$\pmat{U},\pmat{V}\in \RR[t]^{n\times n}$.

One may always write $\Alin+\DeltaAlin$ this way via the SVD for
fixed $\Alin$ and $\DeltaAlin$, so in particular the optimal
solution can be written as a rank factorization.  The problem
$\min \norm{\Alin-UV}^2$ such that $UV$ has the same structure as
$\DeltaAlin$ is generally ill-posed and needs to be constrained to
do any meaningful analysis, as there are numerous degrees of
freedom. At first glance, optimizing over rank factorizations appears
to be a harder problem than the original. However it is helpful to
perform analysis on this formulation. In particular, we are able to
prove that optimal values of $\DeltaAlin$ that satisfy first order
conditions (which contains all useful perturbations) are
separated by a constant amount, and that equivalence classes of
solutions are isolated. Additionally, this formulation of the problem is convex 
in each argument (but not jointly convex) and is amenable to block coordinate 
descent methods. 

We next need to demonstrate that the condition that the matrix
$\Delta\Alin=\Alin-UV$ is the $\RR$-embedding of some matrix polynomial
$\Delta\A\in\RR[t]^\nxn$ can be phrased as a single polynomial being
zero. $\Alin$ is generated by a linear structure $\sum_{i=1}^L c_i \Alin^{(i)}$ 
where $c_i \in \RR$ and $ \{ \Alin^{(1)},\ldots, \Alin^{(L)}\} \subseteq 
\RR^{N\times M}$. Define the structural enforcement function 
\[\DCons:\RR^{N\times R}\times \RR^{R\times M} \to \RR \text{ as } 
\DCons(U,V) = \tallnorm{\sum_{i=1}^L c_i \Alin^{(i)} - \Delta\Alin}_F^2.\] We 
note that there exist $c_i$ such that $\DCons(\Delta\Alin)=0$ if and only if 
$\Delta\Alin$ is an $\RR$-embedding of a matrix polynomial.


\begin{problem}
\label{prb:linear-rank-factor}
With $\Alin, U,V$ as above, the constrained $\RR$-embedded rank
factorization problem consists of computing
  $ \min \norm{\Alin - UV}_F^2$ 
subject to the constraints that $U^T U - I =0$ and $\DCons(U,V)=0$. If
$R= M-1$, then this encodes all rank deficient matrix polynomials.
\end{problem}




%
%
%

It is still not clear that Problem~\ref{prb:linear-rank-factor} is
well-posed, as there are many degrees of freedom in $V$, and this
matrix can have arbitrary rank. The enforcement of $U$ as an
orthogonal matrix ($U^TU-I=0$) is allowed for without loss of
generality. Informally then we are looking at all rank factorizations where
where $U$ is orthogonal and $\DCons(U,V)=0$, that is, the
product satisfies the block-Toeplitz structure on $\DeltaAlin$.



We employ the machinery of non-linear optimization to describe the
geometry of the minimal solutions, and hence the nearest appropriately
structured matrices.  See \citep{Ber99} for an excellent overview.

\begin{fact}[{\citet[Section 3.1.1]{Ber99}}]
  For a  sufficiently large $\rho>0$, one has that\footnote{ $\rho$ is 
sometimes known as a 
penalty term.}  
Problem~\ref{prb:linear-rank-factor} is equivalent to computing a solution to  
the
  unconstrained optimization problem
  \[
    \Phi(U,V)=\min_{U,V} \norm{\Alin-UV}_F^2 + \rho
    \norm{\DCons(U,V)}_F^2 + \rho\norm{U^TU-I}_F^2.
  \] 

\end{fact}

All the interesting solutions to the minimization of $\Phi(U,V)$ occur at
stationary points.  The first-order necessary condition (on $V$) of
gradients vanishing gives us
\begin{align*}
  \nabla_V \left(\norm{\Alin -U V}_F^2 +  
  \rho\norm{\DCons(U,V)}_F^2)+\rho\norm{U^TU-\mat{I}}_F^2 
  \right)
  =0\\
  \iff  U^T(\Alin-UV) +  \left( \frac{\D}{\D 
  V} \DCons(U,V)^T\right)  
  \rho\DCons(U,V)= 0. 
\end{align*}

If we assume that the constraints are active, that is $U$ is
orthogonal and that $\DCons(U,V)=0$, then we have
$U^T\Alin -V=0$. Of course, there is the other first order necessary
condition requiring that
\[\nabla _U \left( \norm{\Alin-UV}^2 + \rho
\norm{\DCons(U,V)}^2 + \rho\norm{U^TU-\mat{I}}^2 \right) 
=0.\] 
However, we do not need to employ this explicitly in the following.

\begin{theorem}[Strong Separation of Objective] \label{thm:strong-sep}
Suppose $\DeltaAlin$ and $\DeltaAlin ^\star$ are  distinct  (local) 
optimal solutions to Problem~\ref{prb:linear-rank-factor} that satisfy first 
order necessary conditions. Then 
$\norm{\DeltaAlin 
- \DeltaAlin^\star}_2 
\geq  \sigma_{\min}(\Alin)$, where $\sigma_{\min}(\cdot)$ is the 
smallest non-trivial singular value.
\end{theorem}
\begin{proof}
From the previously discussed necessary first order condition we have
that there exists $U\in \RR^{N\times R}$, $V\in \RR^{R\times M}$ and 
$U^\star\in \RR^{N\times R^\star}$ 
and $V^\star\in \RR^{R^\star\times M}$ such that
\smallskip
\[
\norm{\DeltaAlin - \DeltaAlin ^\star}_2 =  \norm{U 
V - U^\star V^\star}_2  =  \norm{U U^T \Alin - U 
^{\star} U^{\star T} \Alin}_2. 
\]
Note that $R$ and $R^\star$ need not be the same. 
From this we obtain the sequence of lower bounds
\begin{align*}
 \norm{U U^T \Alin - U^{\star} U^{\star T} \Alin}_2  & \geq \norm{U U^T - 
U^{\star} U^{\star T} }_2  \sigma_{\min}(\Alin)\\
 & =  \norm{\mat{I}- U ^T U^{\star} U^{\star
   T}U}_2\sigma_{\min}(\Alin)  \\
   &\geq \sigma_{\min}(\Alin).
\end{align*}

The symmetric matrix $\mat{W} = U ^T U^{\star} U^{\star T}U$ is a
product of matrices whose non-zero eigenvalues have magnitude
$1$. Symmetric matrices have real eigenvalues, and the non-zero eigenvalues of 
$W$ will be $\pm1$, since $U$ and $U^{\star}$ are orthogonal. Thus 
$\norm{W}_2=1$. 

$\mat{W}$ must have at least one negative eigenvalue or zero 
eigenvalue by the orthogonality
assumption, since $\mat{W}\neq \mat{I}$. Since $\mat{W}$ is symmetric, we 
can diagonalize $\mat{W}$ as a matrix with $\pm1$ and $0$ 
entries on the diagonal. It follows that  
$\norm{\mat{I}-\mat{W}}_2 \geq 1$ and the theorem follows.
\end{proof}
While the separation bound exploited properties of the rank factorization, 
these bounds hold for all formulations of the problem.

\begin{corollary}
  All locally optimal solutions satisfying first order necessary
  conditions 
  are isolated modulo equivalence classes.
\end{corollary}
\begin{proof}
  Suppose the contrary, that is that $(U,V)$ is a solution
  corresponding to $\DeltaAlin$ and $(U^\star, V^\star)$ is a
  solution corresponding to $\DeltaAlin^\star$. The objective
  function and constraints are locally Lipschitz continuous, so let
  $s>0$ be a Lipschitz constant with respect to $\norm{\cdot}_F$ in
  some open neighborhood.

  If we take $0<\varepsilon < \dfrac{\sigma_{\min}(\Alin)}{s}$ then we
  have
  \begin{align*}
    \sigma_{\min}(\Alin) & \leq \norm{\DeltaAlin - \Delta 
                           \Alin^\star}_2 \\
                         & \leq s\tallnorm{ 
                           \begin{pmatrix}
                             U\\ V
                           \end{pmatrix}
    -\begin{pmatrix}
      U^\star \\ V^\star 
    \end{pmatrix}
    }_F  \\ & < \sigma_{\min}(\Alin), 
  \end{align*}
  which is a contradiction to Theorem \ref{thm:strong-sep}.
\end{proof}
Implicitly the matrix $V$ parametrizes the kernel of $\Alin$. If we normalize 
the kernel of $\Alin$ to contain $\RR$-embeddings of primitive kernel vectors 
then the matrix $V$ can be made locally unique, although we do not employ this 
in the  rank-factorization formulation directly.  
%
%
 \begin{corollary}
  Under a suitable choice of $\norm{\cdot}$ we have that minimal solutions are 
separated. In particular, separation holds for $\norm{\cdot}_1$. 
 \end{corollary}
 The proof follows immediately from equivalence of matrix norms, as norms 
are equivalent in a finite dimensional space. 


While there are too many degrees of freedom to easily obtain a
(locally) quadratically convergent minimization over the rank
factorization, the rank factorization does yield non-trivial insights
into the geometry of the solution space. In particular, the isolation
of solutions indicates first order (gradient) methods will perform
well on the problem. In the 
next section we will introduce a locally
quadratically convergent algorithm for an equivalent form of
Problem~\ref{prb:problem-original} that reduces each equivalence class
of solutions to a single \mbox{solution}.
%
%


 \section{An Iterative Algorithm for Lower Rank Approximation}
\label{sec:algorithm}

In this section we propose an iterative algorithm to solve
Problem~\ref{prb:problem-original} based on Newton's method for
constrained optimization.  Sufficient conditions for quadratic
convergence are that the 
second-order sufficiency holds \citep{Wri05} and local Lipschitz continuity of 
the objective and constraints. We ensure these conditions hold for 
non-degenerate problems by
working on a restricted space of minimal $\RR$-embeddings that remove
degrees of freedom.

\subsection{Minimal System of Equations}
In order to compute a nearby rank $n-r$ approximation we want to solve 
the non-convex optimization problem 
\begin{equation} \label{eqn:rank-at-most}
 \min \norm{\Delta \A}_F^2 \text{~~subject to~~} \begin{cases}
                                          (\A+\Delta \A) \pmat{B} =0 \\
                                          \rank(\pmat{B}) = r.
                                         \end{cases}
 \end{equation}
 
 In the instance of (structured) scalar matrices the rank constraint can be 
enforced by ensuring that $\pmat{B}$ has orthogonal columns\footnote{This 
normalization alone is not sufficient for rapid convergence.} or is in a column 
reduced echelon form. In the instance of matrix polynomials this is not 
sufficient, since polynomial multiples of the same vector will have linearly 
independent combined coefficient vectors.  In order to apply these 
normalizations on the 
coefficient vectors of $\pmat{B}$ we require that the columns be represented 
with a minimal number of equations with respect to $\pmat{B}$.


\begin{definition}[Minimal $\RR$-Embedding]
%
  Suppose 
  $\A \in \RR[t]^{n\times n}$ with $\RR$-embedding $\Alin$. The
  vector $\bb \in \RR[t]^{n\times 1}$, with $\RR$-embedding $\blin$, is
  said to be {\em minimally $\RR$-embedded} in $\Alin$ if
  $\ker \Alin = \langle\blin\rangle$ (i.e., a dimension 1 subspace).
%
%
%
  We say that $\blin$ is {\em minimally degree $\RR$-embedded} in $\Alin$ if
  (1) $\blin$ is minimally $\RR$-embedded in $\Alin$ and (2)
$\blin$ corresponds to a primitive kernel vector $\pmat{b}$, that is
  $\gcd( \pmat{b}_{1}, \ldots, \pmat{b}_{n})=1$.
%
\end{definition}

We note that this definition ensures minimally $\RR$-embedded
vectors are unique (up to scaling a factor), or that $(\Alin_j 
+\DeltaAlin_j)\Blin[*,j] 
=0$ has a
(locally) unique solution for fixed $\DeltaAlin$.  In the minimal embedding, we 
will assume, without loss of generality, that redundant or equations known in 
advance, such as $0=0, \Delta\Alin_{ij}=0$ or $\Blin_{ij}=0$ corresponding to 
known entries are removed 
for some indices of $i$ and $j$. 
Some of these trivial equations occur because 
of the CREF assumption, while others occur from over-estimating degrees of 
entries.

This allows us to reformulate $(\A+\Delta \A) \pmat{B}=0$ 
as a (bi-linear) system of equations 
\begin{equation} \label{eqn:minimal-system}
\{ (\Alin_j + \Delta \Alin_j)\Blin[*,j]=0\}_{j=1}^{r}
\end{equation}
where the $j^{th}$ column of $\pmat{B}$ is 
minimally degree embedded in the system $(\Alin_j+\Delta \Alin_j)$.   
We also 
note that assuming $\pmat{B}$ is in a column-reduced echelon form essentially 
requires us to guess the pivots in advance of the optimal solution, which is 
only possible with a good initial guess. The benefit of this approach is that 
if the pivots are not guessed correctly, we are still able to compute a $n-r$ 
approximation of $\A$.

In order to exclude trivial solutions, we can assume that the pivot elements of 
$\pmat{B}$ have a norm bounded away from zero.
Let $\N(\blin_i)$ be a 
normalization vector such that $\N(\blin_i)^T \blin_i =1$ which implies that 
the CREF pivots are bounded away from zero. For example, take the pivot to 
have unit norm.
Note that other normalization vectors are possible, such as $\N(\blin_i) = 
\blin_i$ (which corresponds to each column having a unit norm) if the initial 
guess is adequately close, or we could take the pivot element to be a monic 
polynomial.   Of course there are several other permissible normalizations.

Define the matrix $\Alin_i$ to have  the column $\blin_i =\Blin[1..n,i]$ 
minimally  degree embedded. 
We can express  
\eqref{eqn:minimal-system} in a vector-matrix form as follows.

\begin{equation} \label{eqn:rank-at-most-system}
\begin{pmatrix}
   \Alin_1 + \DeltaAlin_1 \\
   & \Alin_2+\DeltaAlin_2 \\
   && \ddots \\
   &&&  \Alin_r+\DeltaAlin_r \\ \hline 
   \N(\blin_1)^T \\
   & \N(\blin_2)^T \\
   && \ddots \\
   &&& \N(\blin_r)^T
  \end{pmatrix} 
  \begin{pmatrix}
  \blin_1  \\ \blin_2  \\ \vdots \\ \blin_r
  \end{pmatrix} = 
  \begin{pmatrix}
      0 \\
      0\\
      \vdots \\
      0 \\ \hline
      1 \\
      1\\
      \vdots \\
      1
  \end{pmatrix}
\end{equation}
 has a (locally) unique solution for fixed $\Delta\A$.

\subsection{Lagrange Multipliers and Optimality 
Conditions}\label{ssec:optimality}
In order to solve \eqref{eqn:rank-at-most} we will use the method of Lagrange 
multipliers \citep{Ber99}.

Let $M(\Delta \A, \B)$ 
be the vector of residuals corresponding to \eqref{eqn:rank-at-most-system}, 
then the Lagrangian is defined as 
\begin{equation}\label{eqn:lagrangian}
 L = \norm{\Delta \A}_F^2 + \lambda ^T 
M(\Delta \A, \B),
\end{equation}
where $\lambda = (\lambda_1,\ldots,\lambda_{\text{\# residuals}})^T$ is a 
vector of Lagrange multipliers.

\begin{definition}
 The vectorization of $\pmat{A}\in \RR[t]^{n\times n}$ of degree at 
most $d$ is defined as 
 \[ 
   \vec(\pmat{A}) =  (\pmat{A}_{1,1,0},\ldots,\pmat{ A}_{1,1,d}, 
\pmat{A}_{2,1,0}, \ldots, \pmat{A}_{2,1,d}, \ldots, 
\pmat{A}_{n,n,0}, \ldots \pmat{ A}_{n,n,d})^T,
\]
that is $\vec(\pmat{A}))$ stacks the entry-wise coefficient vectors of each 
column on top of each other. 
\end{definition}
We will find it convenient to define 
    $x= \begin{pmatrix}
         \vec (\Delta \A) \\
         \vec (\pmat{B})
        \end{pmatrix}$ to be the combined vector of unknowns corresponding to 
$\A$ and $\B$. 
Let $\nabla^2_{xx}L$ denote the Hessian matrix of $L$ with respect to $x$ and 
$J$ be the Jacobian of the residuals of the constraints, i.e. $J = \nabla_x 
M(\Delta \A,\pmat{B})$.  
Necessary optimality conditions at a point $(x^*,\lambda^*)$  \citep{Ber99} are 
that 
\begin{equation} \label{eqn:necessary-conditions} 
\nabla L= 0 \text{ and } \ker(J)^T\nabla_{xx}^2 L \ker(J) \succeq 0.
\end{equation}
Sufficient conditions for optimality at the same point are that
\begin{equation} \label{eqn:sosc}
\nabla L= 0 \text{ and } \ker(J)^T \nabla_{xx}^2 L \ker(J) \succ0.
\end{equation}
These conditions are known as the second-order sufficiency conditions 
\cite{Ber99}.
We note that \eqref{eqn:sosc} implies that minimal solutions are locally 
unique, and will fail to hold if minimal solutions are not locally unique. The 
idea is to show that \eqref{eqn:sosc} holds in the minimal embedding, which 
allows us to construct an algorithm with rapid local convergence.


\subsection{The Jacobian}

\begin{definition}
  The matrix $\psi(\blin)$ is an alternative form of
  $(\Alin+\DeltaAlin) \blin=0$ that satisfies
  $\psi(\blin) \vec(\A+\Delta \A)=0$.  That is, $\psi(\blin)$
  satisfies
\[
\psi(\blin) \cdot \vec (\A + \Delta \A) = 0 \iff
(\Alin + \DeltaAlin) \blin =0.
\]
\end{definition}
 
We will adopt that notation that $\psi(\blin_i)$ corresponds to 
$\psi(\blin_i) \vec(\Alin_i+\Delta \Alin_i) = 0$.
Here we use the bi-linearity of \eqref{eqn:rank-at-most-system} to 
write 
the same system using  a matrix with entries from 
$\Blin$ instead of $\vec(\A+\Delta \A)$. 

The closed-form expression for the Jacobian of the residuals (up to 
permutation) in 
\eqref{eqn:rank-at-most-system} is given by 
\begin{equation}
J  = \left( \begin{array}{c|cccc}
          \psi(\blin_{1}) & \Alin_1+\DeltaAlin_1 & & &  \\
          \psi(\blin_{2}) & & \Alin_2+\DeltaAlin_2  & &  \\
          \vdots  & &  & \ddots &  \\
          \psi(\blin_{r}) & && & \Alin_r+\DeltaAlin_r      \\ \hline
          0& \N(\blin_{1})^T & & &\\
          0&& \N(\blin_{2})^T  && \\
          \vdots&& & \ddots &   \\
          0&& & &  \N(\blin_{r})^T 
         \end{array} \right).
\end{equation}
Unlike the case of a single kernel vector in \citep{GHL17}, $J$ may be rank 
deficient since some 
equations corresponding to low (high) index entries may be redundant 
at the solution. The Lagrange multipliers will not be unique in this particular 
scenario and the rate of convergence may degrade if Newton's method is used. In 
the instance of $r=1$ then we present the following result \citep{GHL17}.

\begin{theorem}\label{thm:jacobian-rank}
  Suppose that $r=1$ and $\blin_1$ is minimally degree $\RR$-embedded in 
$\Alin_1$, then
  $\mat{J}$ has full rank when \eqref{eqn:necessary-conditions} holds.
\end{theorem}
\begin{proof}

We show that $J$ has full row rank by contradiction.
If this matrix was rank deficient, then one row is a linear 
combination of the others.
This means that one of the equations in the constraints is trivial or the 
solution is not regular (see \cite[Section 3.1]{Ber99}). As we are only 
concerned about regular solutions, this 
contradicts the minimal degree $\RR$-embedding.
\end{proof}
The corollary to this is that in the minimal embedding regularity conditions 
hold and it is straight forward to obtain rapid local convergence.

\subsection{The Hessian}
The Hessian matrix, $\nabla^2L$ is straight forward to compute as 
\[ \nabla^2 L = \begin{pmatrix}
                 \nabla^2_{xx}L & J^T\\
                 J & 0 
                \end{pmatrix}. 
\]
The following theorem shows that second-order sufficiency holds for the 
instance of $r=1$.  The case of $r>1$ follows immediately by induction. This is 
in contrast to Theorem~\ref{thm:jacobian-rank}, which does not always hold for 
$r>1$. 

\begin{theorem}[Second Order Sufficiency Holds]\label{thm:second-order-suff}
 Suppose that $\Alin+\DeltaAlin$  has a minimally degree $\RR$-embedded 
kernel 
vector $\blin$, i.e. $r=1$ in \eqref{eqn:lagrangian}, then at a minimal 
solution, the second order 
sufficiency condition \eqref{eqn:sosc} holds in 
the minimal embedding of $\blin$. 
\end{theorem}
\begin{proof}
 If $\norm{\Delta A}=0$ at the local minimizer $(x^*,\lambda^*)$ then 
\[\nabla^2_{xx}L(x^\star,\lambda^\star) = 
 \begin{pmatrix}
  2I \\
  & 0
 \end{pmatrix} \text{ and }
K=\ker 
\nabla^2_{xx}L(x^\star,\lambda^\star) = \Span \begin{pmatrix}
                                                   0 \\ 
                                                   & I
                                                  \end{pmatrix}.\] 
                                                  
We have that for  $y\in \Span(K)$ such that $Jy=0$  implies that $\Alin y=0$ 
and $\N(\blin)^T y=0$. It follows that $\ker \Alin = \Span(\blin)$, thus we 
have 
$y=\blin$ or $y=0$  via the 
minimal degree $\RR$-embedding,  thus $y=0$ as $\blin \notin \Span(K)$. Hence, 
second-order  sufficiency holds, as $\ker J \cap K = 0$.
                                                  
If $\norm{\Delta\A}\neq 0$ then we have that 
\[ \nabla^2_{xx}L(x^\star,\lambda^\star) = \underbrace{\begin{pmatrix}
                                            2I&0 \\
                                            0& 0
                                           \end{pmatrix}}_{\mathcal{H}} +        
          \underbrace{\begin{pmatrix}
             0 & E^T \\
             E & 0
            \end{pmatrix}}_{\mathcal{E}}.\]
The matrix $\mathcal{E}$ is linear in $\lambda$, however the precise tensor 
decomposition is irrelevant to the proof. If $E$ has full rank, then 
$\nabla^2_{xx}L$ has full rank and we are done, so suppose  that 
$E$ is rank 
deficient. If $E$ is rank deficient, then one can eliminate a row of $E$ and 
column of $E^T$ without affecting  $\mathcal{H}$ via 
symmetric row and column updates. We observe that 
$\ker(\mathcal{H}+\mathcal{E}) \subseteq \ker \mathcal{H}$ and the result 
follows.
\end{proof}
\begin{corollary}
Suppose that $r>1$ in \eqref{eqn:lagrangian} and $\B$ is minimally degree 
embedded, then second-order sufficiency 
\eqref{eqn:sosc} holds.
\end{corollary}
\begin{proof}
 The proof is almost the same as Theorem~\ref{thm:second-order-suff} and 
follows by induction on $r$ since each block is  
decoupled.
\end{proof}

We now have all of the ingredients for an iterative method with rapid local 
convergence. 

\subsection{Iterative Post-Refinement}\label{ssec:post-refinement}

Newton's method for equality constrained minimization problems can be 
interpreted as solving the non-linear system of equations $\nabla L=0$. 
Newton's method is based on the iterative update scheme 

\begin{equation} \label{eqn:newton-step}
\begin{pmatrix}
   \mat{x}^{k+1} \\
   \lambda^{k+1}
  \end{pmatrix}
  =
  \begin{pmatrix}
  \mat{x}^{k} +\Delta \mat{x}^k\\
   \lambda^{k} +\Delta {\lambda}^k
  \end{pmatrix}
\text{ such that }
 \nabla^2 L  \begin{pmatrix}
                \Delta \mat{x}\\
                \Delta \lambda                         
		\end{pmatrix}.
\end{equation}
If $r=1$ then $\nabla^2L$ has full rank and the iteration is well defined by
matrix inversion. If $r>1$ then we consider the quasi-Newton method 
defined as
\begin{equation}\label{eqn:reg-system}
\begin{pmatrix}
 x^{k+1} \\
 \lambda^{k+1} \\
\end{pmatrix} = 
\begin{pmatrix}
 x^k + \Delta x^k \\
 \lambda ^k + \Delta \lambda ^k 
\end{pmatrix}
\text{ such that }
\begin{pmatrix}
   \nabla^2_{xx} L & J^T \\
   J & -\mu_k I
  \end{pmatrix}  
  \begin{pmatrix}
   \Delta x \\
   \Delta \lambda
  \end{pmatrix}
  = -\nabla L
\end{equation}
for a suitably chosen parameter $\mu_k$.  Taking $\mu_k = \norm{\nabla 
L{(x^{k},\lambda^k)}}_1$ one has provably quadratic convergence 
\citep[Theorem~4.2]{Wri05} with $x^k$ and $\lambda^k$ chosen 
sufficiently close to the 
optimal solution.

\begin{theorem}
The iteration 
\eqref{eqn:reg-system} converges quadratically  to $(x^\star,\lambda^\star)$ if 
$(x^0,\lambda^0)$ are chosen sufficiently close to $(x^\star,\lambda^\star)$.
\end{theorem}

We now have a method to compute a nearby rank deficient matrix polynomial with 
a rate of convergence that is quadratic, provided that the initial values of 
$x$ are chosen to be sufficiently close to the optimal solution.

\section{Implementation, Description and Comparison}
\label{sec:implementation}

In this section we discuss implementation details and demonstrate our
implementation for computing the nearest rank deficient matrix
polynomial.  All algorithms are implemented in Maple 2016. %
All experiments are
done using quad precision floating point arithmetic, with about $35$
decimal digits of accuracy. We compare some degree one examples to the recent 
results of \citep{GulLubMeh16}.

To compute an approximate kernel vector, 
first we use the SVD to compute an approximate kernel of an $\RR$-embedded 
(nearly) rank 
deficient matrix polynomial. 
Next we use  
structured orthogonal elimination $RQ$  ($LQ$) decomposition to produce a 
minimally (degree) $\RR$-embedded vector from the kernel. In the case of 
several kernel vectors we use a modified Gaussian elimination on an 
embedding of an approximate kernel obtained by the SVD and approximate GCD to 
find nearby approximate kernel vectors that are primitive.

\subsection{Description of Algorithm}
We now formally describe an algorithm for computing the nearest matrix 
polynomial of a prescribed rank. The algorithm has no global convergence 
guarantees, however 
a globally convergent (although not necessarily optimal) algorithm can be 
developed in a straight forward manner via augmenting our second order 
algorithm with a first order one, and removing content from kernel vectors if 
necessary.

\begin{algorithm}[!h] \label{alg:newton-kernel}
  \caption{\textbf{:} \texttt{Iterative Kernel Post-Refinement}}
  \label{alg:refinement}
 \begin{algorithmic}[1]
   \smallskip 
   \Require
 \item [$\bullet$] Full rank matrix polynomial $\pmat{A} \in \RR[t]^{n\times n}$
 \item [$\bullet$] (Approximately) Rank deficient matrix polynomial $\pmat{C} 
\in 
\RR[t]^{n]\times n}$
 \item [$\bullet$] Approximate kernel vectors $\pmat{c}_1,\ldots, \pmat{c}_r 
\in \RR[t]^{n\times 1}$ 
of the desired degree/displacement structure
 \item [$\bullet$] Displacement structure matrix $\Delta \pmat{A}$ to optimize 
over
   \smallskip \Ensure
 \item[$\bullet$] Singular matrix $\pmat{A}+\Delta \pmat{A}$ with $\B \subset 
\ker (\pmat{A}+\Delta \pmat{A})$ or an indication of failure.
 
   \medskip \State $\RR$-Embed $\pmat{A}, \pmat{C}, \pmat{c}_1,\ldots, 
\pmat{c}_r$ and $\Delta \pmat{A}$.
\State Compute Lagrangian $L$ from Section~\ref{ssec:optimality}.
\State Initialize $\lambda$ via linear least squares  from 
$\nabla L|_{\mat{x}}=0$.
\State Compute $\begin{pmatrix}
                 \mat{x}+\Delta \mat{x} \\
                 \mat{\lambda}+\Delta \mat{\lambda}
                \end{pmatrix}$ by solving \eqref{eqn:reg-system} until 
$\tallnorm{\begin{pmatrix}
        \Delta \mat{x}\\
        \Delta \mat{\lambda}
       \end{pmatrix}
}_2$ is sufficiently small or divergence is detected.

\State Return the locally optimal $\Delta \pmat{A}$ and $\B$ or an 
indication of failure.
 \end{algorithmic}
\end{algorithm}

The size 
of $\nabla^2 L$ is $O(r^2n^4 d^2)$ and accordingly each iteration has a cost of 
$O(r^6n^{12}d^6)$ flops using standard matrix multiplication, where $r$ is the 
dimension of the kernel.

\subsection{Nearest Rank Deficient Linearly and Affinely Structured Matrix}
In this section we consider Examples 2.10, 2.11 and 2.12 from 
\cite{GulLubMeh16}, where we compare our results to real 
perturbations. Note 
that complex perturbations are a straight-forward generalization of the theory 
presented here, and can be re-formulated as a problem over $\RR$. 

The technique of \cite{GulLubMeh16} poses computing a nearby rank-deficient 
linear matrix pencil by verifying that sufficiently many images of the matrix 
polynomial are singular, so that $\det(\A+\Delta \A)\equiv0$. The problem is 
then 
posed as a solution to a system of Ordinary Differential Equations (ODE), 
assuming 
that certain genericity conditions on the eigenvalues of the solution 
hold\footnote{Our algorithm and convergence theory does not explicitly rely on 
genericity 
assumptions or other properties of eigenvalues, however we do exploit 
generic properties 
in formulating initial guesses.}.  They consider the instances of computing 
$A_0$ and $A_1$ with a common kernel vector, and the instance where $A_0$ and 
$A_1$ do not have a common kernel. Additionally, perturbations affecting only 
one of $A_0$ and $A_1$ are considered.
We note that the solutions to 
the ODEs do not necessarily satisfy necessary optimality conditions 
\eqref{eqn:necessary-conditions}, and accordingly will generally not be local 
minimizers.

 \subsubsection{Nearest Affinely Structured Examples I}
 Consider first the matrix polynomial 
\[\A =  \underbrace{\begin{pmatrix} 0&0&0\\ 
 0&0&1\\ 
 0&1&0\end{pmatrix}}_{A_1} t  + 
\underbrace{ \begin{pmatrix} 0& 0.0400& 0.8900\\   
0.1500&- 0.0200&0\\   0.9200& 0.1100& 0.06600\end{pmatrix}}_{A_0}  
 \]
coming from Examples 2.10 and 2.12 of \cite{GulLubMeh16} 

\begin{example}
 If we assume that $A_1$ is constant, then this is finding the (locally) 
nearest matrix polynomial with an affine structure since $A_1$ has non-zero 
fixed constants. First let's assume that zero entries are preserved, this is a 
linear structure on $A_0$.

To compute an initial guess for $\b$ we use the SVD on $\Alin$ and extract a 
guess from the smallest singular vector. This gives us 
\[\b_{init} =  \begin{pmatrix} - 0.41067 {t}^{3}+ 
0.50576 {t}^{2}- 0.26916 t- 0.035720\\    
0.38025 {t}^{2}- 0.51139 t+0.30674\\    0.027012 {t}^{2}-
 0.028083 t+ 0.010715\end{pmatrix}.\]
For an initial guess on $\A$ we take $\A_{init}=\A$.  Note that we do not need 
an initial guess that is singular, it just needs to be ``sufficiently close'' 
to a singular matrix polynomial.

If we do not allow perturbations to zero-coefficients, that is, $A_0[1,1]$ and 
$A_0[2,3]$ may not be perturbed, then 
after five iterations of plain Newton's method (see \citep{GHL17}) we compute
\[ \Delta A_0 \approx \begin{pmatrix}  0.0&- 0.094149&- 
0.0057655\\  - 0.093311& 0.026883& 0.0\\   
0.0057142&- 0.0016462&- 0.00010081\end{pmatrix}\] with perturbation 
 $\norm{\Delta\A}_F \approx 0.135507$.
 
 A corresponding (approximate) kernel vector is 
 \[ \b \approx \begin{pmatrix}  0.73073 t+ 0.082126\\ 
 - 0.67644\\  - 0.041424\end{pmatrix}.
\]
\end{example}
\begin{example}
If we allow perturbations to zero-coefficients in $A_0$ then after five rounds of 
plain Newton's method we compute 
\[ \Delta A_0 \approx 
  \begin{pmatrix} 0.0&- 0.094179&- 0.0057705\\ 
 - 0.093280& 0.026786& 0.0016412\\   
0.0057154&- 0.0016412&- 0.00010056 
\end{pmatrix} \] with perturbation  $\norm{\Delta \A}_F 
\approx 0.135497$,
which is a marginal improvement over the previous example. A corresponding 
approximate kernel vector is 
\[ \b \approx \begin{pmatrix}  0.73073 t+ 0.082131\\ 
 - 0.67644\\  - 0.041447\end{pmatrix}.
\]
\end{example}

\cite{GulLubMeh16} report an upper-bound on the distance to 
singularity allowing {\emph{complex perturbations}}, that is $\Delta \A \in 
\CC[t]^{n\times n}$ of $\norm{\Delta^{\CC} \A}_F \approx 0.1357$ in Example 
2.10. In Example 2.12, \cite{GulLubMeh16} report an upper-bound on the distance 
to singularity
 allowing \emph{real perturbations},  $\norm{\Delta^{\RR}\A}_F\approx 
0.1366$.
Although we only consider real perturbations, both bounds are improved. We 
conjecture that the complex  bound can be improved further.


If we allow  perturbations to $A_0$ and $A_1$, then this is some form of 
finding the nearest 
rank deficient matrix polynomial. The question is whether to allow degree or 
support preserving perturbations. Again, we will use the same initial guesses 
as the previous example.

Matrix degree preserving perturbations are of the form \[\Delta^{deg} \A 
=\begin{pmatrix} 
tA_{{1,1,1}}+A_{{1,1,0}}&tA_{{1,2,1}}+A_{{1,2,0}}&tA_{{1,3,1}}+A_{{1,3,0}}\\ 
 tA_{{2,1,1}}+A_{{2,1,0}}&tA_{{2,2,1}}+A_{{2,2,0}}&tA_{{2,3,1}}
 +A_{{2,3,0}}
\\ 
 tA_{{3,1,1}}+A_{{3,1,0}}&tA_{{3,2,1}}+A_{{3,2,0}}&tA_{{3,3,1}}
+A_{{3,3,0}}\end{pmatrix},
 \]
 where  as support preserving perturbations are of the form 
 \[ \Delta^{sup}\A = \begin{pmatrix} 0&A_{{1,2,0}}&A_{{1,3,0}}\\ 
 A_{{2,1,0}}&A_{{2,2,0}}&A_{{2,3,1}}t\\ 
 A_{{3,1,0}}&tA_{{3,2,1}}+A_{{3,2,0}}&A_{{3,3,0}} 
\end{pmatrix}.
 \]
 
 \begin{example}\label{ex:arbitrary-degree}
 In the instance of degree preserving perturbations we compute after five 
iterations of Newton's method
 \[ \Delta^{deg}\A \approx \begin{pmatrix}  0.0036502& 0.0039174 t- 
0.066405& 0.00011839 t- 0.0020069\\ 
 - 0.066897& 0.058993 t+
 0.029807& 0.0017829 t+ 0.00090082\\   0.0059893&- 
0.0053098 t- 0.0024133&- 0.00016047 t- 0.000072934\end{pmatrix} 
 \]
 with $\norm{\Delta^{deg}\A} \approx 0.115585$.
 
 A corresponding approximate kernel vector is 
 \[\b \approx  \begin{pmatrix} - 0.72941 t- 0.080355\\ 
  0.67903\\   0.020522\end{pmatrix}   .
\]
 \end{example}
 \begin{example}
 In the instance of support preserving we compute after five iterations of 
Newton's method, 
 \[ \Delta^{sup}\A \approx   \begin{pmatrix}  0.0&- 0.094311&- 
0.0057928\\  - 0.092552& 0.026973& 0.0051028 t\\ 
  0.0057434&- 0.0051554 t- 0.0016739&- 0.00010281 
\end{pmatrix} \]  with $\norm{\Delta ^{sup} \A} \approx 0.135313.$
 A corresponding approximate kernel vector is 
 \[ \b \approx  \begin{pmatrix} - 0.72895 t- 0.082339\\ 
  0.67832\\   0.041664\end{pmatrix}.
\]
\end{example}

\cite{GulLubMeh16} report an upper-bound on the distance to 
singularity of $\norm{\Delta^{deg}\A}_F\approx 0.1193$ in Example 
2.12. This bound is larger than  the one computed in 
Example~\ref{ex:arbitrary-degree}.

\subsection{Nearest Affinely Structured Examples II}

 \begin{example}
 Next we consider the the matrix polynomial $\A$ in Example 2.11 of
\citep{GulLubMeh16} defined as
\[\A = \underbrace{\begin{pmatrix} - 1.79& 0.10&- 0.6\\ 
  0.84&- 0.54& 0.49\\  - 0.89& 0.3& 0.74\end{pmatrix}}_{A_0} + 
 \underbrace{\begin{pmatrix} 0&0&0\\  0&0&1\\ 
 0&1&0\end{pmatrix}}_{A_1}t.
 \]
 To compute an initial guess for we take $\A_{init}=A$ and take \[\b^{init} =  
 \begin{pmatrix} - 0.16001 {t}^{3}- 
0.10520 {t}^{2}+ 0.15811 t+ 0.11409
\\  0.14980 {t}^{3}- 0.51289 {t}^{2}- 0.18616 t+ 0.54098
\\ 0.20801 {t}^{3}+ 0.26337 {t}^{2}- 0.44619 t- 
0.027979\end{pmatrix}. \]
$\b^{init}$ is computed from the smallest singular vector of $\Alin$.

We note that this initial guess does not attempt to find a nearby singular 
matrix polynomial for the initial guess, all that is needed is $\nabla 
L(\mat{x}^{init},\lambda^{init})$  is reasonably small to obtain convergence.

Using a 
globalized variant of Newton's method based on Levenberg-Marquardt we 
compute 
 \[\Delta \A= \begin{pmatrix}  0.047498 t+ 0.17772& 0.44989 t+ 
0.12420&- 0.091945 t- 0.068210\\   0.20979 t+ 0.078872&- 
0.094205 t+ 0.41583&- 0.037916 t- 0.094081
\\   0.082862 t- 0.15413&- 0.58334 t+ 0.12940& 0.081637 t+ 
0.017208\end{pmatrix}, \]
 $ \text{ with } \norm{\Delta \A}_F \approx  
0.949578.$ The corresponding approximate kernel vector is 
 \[\b=   \begin{pmatrix} - 0.29258 t- 0.21491\\   
0.044825 t- 0.90281\\   0.068189 t+ 0.21562  
\end{pmatrix}.
\]

If we use the result of \citep{GulLubMeh16} as the initial guess,
then we compute 
\[\b^{init} = 
\begin{pmatrix}  0.16409 {t}^{2}+ 0.25146 t+ 
0.12362\\ 
-4.5353\times 10^{-14} {t}^{2}+ 0.23740 t+ 0.55516\\
 1.2457\times 10^{-13} {t}^{2}- 0.48688 t- 0.0060443\end{pmatrix}.
 \]
 We will assume the entries of $\b$ are degree at most two. 
 
 After five iterations of Newton's method we obtain 
\[ \Delta \A =  \begin{pmatrix}  0.17257& 
0.12237 t+ 0.25225&- 0.46902 t+ 0.087147\\   0.21449& 
0.15210 t+ 0.31353&- 0.58296
 t+ 0.10832\\  - 0.055963&- 
0.039685 t- 0.081803& 0.15210 t- 0.028261\end{pmatrix},
 \]
with $\norm{\Delta \A} \approx 0.94356416.$

The corresponding approximate kernel vector is 
\[ \b = \begin{pmatrix}  0.18971 {t}^{2}+ 0.29750 t+ 0.14667\\ 
  0.27896 t+ 0.66186\\  - 0.58143 t- 
0.0079694\end{pmatrix} .
\]
The previously noted small quadratic terms were at roughly machine precision 
(the computation is done with 35 digits of precision) and truncated.
\end{example}
\cite{GulLubMeh16} obtain a result on this past example that  produces an 
upper  bound on the distance to singularity of $0.9438619$. Their computation 
is  accurate to seven decimal points, and accordingly our  post-refinement has 
an  improvement of about $0.000297$. This is not surprising, since we solve the 
necessary conditions \eqref{eqn:necessary-conditions} directly with a 
reasonable initial guess.

\subsection{Lower Rank Approximation of a $4\times 4$ Matrix}
In this following example we consider computing a lower-rank approximation to a 
given matrix polynomial. Consider the $4\times 4$ matrix  polynomial $\A$,
defined as \[\A = A_0 + A_1 t + A_2 t^2+A_3t^3, \text{ where }\]

\begin{align*}
A_0 & =     \begin{pmatrix}  0.09108776&- 0.05442464& 0.3645006& 
0.01821543\\    -
 0.1456436& 0.03647524&- 0.07277662& 0.07305016\\     
0.05478714&- 0.05444916&
 0.4373220& 0.05478385\\    - 0.1274211& 0.09124859&- 0.6556615&- 
0.05446850
\end{pmatrix},
\\
A_1 & =  \begin{pmatrix}  0.09116729& 0.00001797690& 0.2550857& 
0.05475106
\\     0.0001156514& 0.00001659159& 0.09108906&- 0.05447104\\
 0.05470823& 0.03662426& 0.1276959& 0.03650378\\     0.05472202&- 
0.1091389&
 0.1458359&- 0.09090507\end{pmatrix},
\\
A_2 & =    \begin{pmatrix}  0.01833149& 0.03661770& 0.01824331& 
0.03660918\\    
 0.01837542&- 0.05442525& 0.0& 0.01832234\\     0.01841784& 
0.00003900436& 0.0&
 0.01836515\\     0.01840752& 0.00001508311& 0.01839699& 
0.03659170 
\end{pmatrix} ,
\\
A_3 & =    \begin{pmatrix}  0.0& 0.01837967& 0.0& 0.0\\ 
    0.0& 0.01843603& 0.0&
 0.0\\     0.0& 0.01829203& 0.0& 0.0\\     0.0& 
0.01842778& 0.0& 0.0
\end{pmatrix}.
\end{align*}

\begin{example}
We will consider a displacement structure on the kernel as well in this 
example, where higher-order zero terms are not perturbed from the initial 
guess.
For the entries of $\Delta \A$ we preserve higher-order zero terms, and 
allow low-order terms to be perturbed. This is a linearly structured problem, 
on both the main variable $\Delta \A$ and the auxiliary kernel variable $\B$. 

To ensure the rank constraint holds, we will additionally assume that the 
kernel, $\Blin$ 
is in a CREF (while $\B$ is obviously not) and the columns have unit norm.  
This normalization is (locally) equivalent to the ones discussed in 
Section~\ref{ssec:optimality}. Having $\Blin$ in a CREF ensures that the two 
kernel vectors are locally linearly independent during the iteration. Of course 
perturbing both pivots to zero is possible (although this is sub-optimal). In  
 such a scenario linear independence can no longer be guaranteed, and the 
iteration would need to be re-ininitialized. 

For the initial guess we use $\A^{init}=\A$ and take $\B^{init}$ as 
\[ \scalemath{.8}{   \begin {pmatrix} 0.1954059 {t}^{2}& 0.0\\ 
 - 0.2526800 t- 0.7681472&- 0.06131396 {t}^{2}-
 0.1839419 t+ 0.7357675\\  - 0.05727413 {t}^{2}- 
0.01010720 t- 0.1280246&- 0.06131396 {t}^{3}- 0.06131396 t+ 0.1226279
\\   0.05727413 {t}^{2}+ 0.4683004 t+ 0.2560491& 
0.06131396 {t}^{3}+ 0.4905117 {t}^{2}- 0.3065698 t- 0.2452558  
\end{pmatrix}}.  \]

Using Algorithm~\ref{alg:newton-kernel} we  compute after nine iterations  
\begin{align*}
 \Delta A_0 & =  \begin{pmatrix}  0.00003841866&- 0.0001970606&- 
0.00002444167&- 0.000003273264
\\     0.00001831140&- 0.00009026377& 0.00002067189&- 
0.0001255102
\\    - 0.0001265513&- 0.0001595407& 0.00003425737&- 
0.00007523197
\\    - 0.00007666528&- 0.0002773970& 0.00004057408&- 
0.0001720881 
\end{pmatrix} , \\
 \Delta A_1 & =  \begin{pmatrix}  0.00001508776& 0.00003166597& 
0.00004647888&- 0.0001142308
\\    - 0.00005872595&- 0.00004487730& 0.00004547421&- 
0.0001483973
\\     0.00002056901&- 0.0001596527&- 0.000006413632&- 
0.00006541721
\\    - 0.00003695701&- 0.0001773889& 0.00004119722&- 
0.0002159825 
\end{pmatrix} ,
\\
 \Delta A_2 & = \begin{pmatrix} - 0.00003352295&- 0.0001190577& 
0.00005687700&- 0.0001783770
\\     0.00001768442&- 0.0001467423& 0.0&- 0.00008587235\\ 
   -0.00006506345& 0.00005243135& 0.0&- 0.0001686619\\    - 
0.0001471227&-
 0.0001295490&- 0.00001105246&- 0.0001124559\end{pmatrix},
 \\
 \Delta A_3 & =  \begin{pmatrix}  0.0&- 0.0001025690& 0.0& 0.0\\ 
    0.0&- 0.0001315095
& 0.0& 0.0\\     0.0&- 0.00002763942& 0.0& 0.0\\ 
    0.0&-
 0.0001877673& 0.0& 0.0\end{pmatrix},
\end{align*}
with $\norm{\Delta \A} \approx 0.0007844$. 

An approximate kernel, $\B$ is given by 
\[ \scalemath{.8}{    \begin{pmatrix} 0.1955493 {t}^{2}+ 
0.0006874986 t- 0.001013023& 0.0\\  - 0.2542383 t- 
0.7686061&- 0.06128819 {t}^{2}- 0.1818298 t+ 0.7368313
\\  - 0.05698735 {t}^{2}- 0.01004111 t- 0.1276311&- 
0.06125293 {t}^{3}- 0.0002486115 {t}^{2}- 0.06112324 t+ 0.1226783\\ 
  0.05795811 {t}^{2}+
 0.4677475 t+ 0.2541290& 0.06151690 {t}^{3}+ 0.4894569 {t}^{2}- 0.3069667 t- 
0.2452396\end{pmatrix}}. \]
\end{example}

A natural question is what happens if we change the displacement structure on 
the kernel? To investigate this behavior, we consider an equivalent 
representation of the previously used kernel, except that $\B$ is in a CREF 
directly.
\begin{example}
If we change the kernel $\B^{init}$  to be 
\[\scalemath{.8}{   \begin{pmatrix} 0.1581139 {t}^{3}+ 
0.1581139 t- 0.3162278& 0.03965258 {t}^{3}+ 0.3172206 {t}^{2}- 0.1982629 t- 
0.1586103\\  - 0.1581139 {t}^{2}-
 0.4743417 t- 0.6324556&- 0.03965258 {t}^{2}- 0.4361784 t- 0.7930516\\ 
  0.0& 0.07930516 t- 0.07930516\\   
0.3162278 t- 0.3162278& 0.0 
\end{pmatrix}},\] used in the initialization of the previous example, 
 then we compute a perturbation with
$\norm{\Delta \A} \approx 0.0008408$.  
\end{example}

In either case, we obtain comparable answers that are a reasonable lower-rank 
approximation, and can likely be improved by relaxing restrictions on the 
displacement structure on $\B$ or $\Blin$. It is important to note that 
relaxing the degree bounds to be $(n-r)d$ in general on all non-zero entries 
(where entries are zero if they are in the same row as a CREF pivot) will 
likely lead to a better approximation, however one may lose quadratic 
convergence if doing so, since iterates may no  longer have primitive kernel 
vectors, and \eqref{eqn:sosc} will no longer hold.  As discussed in 
Section~\ref{sec:algorithm}, it is generally difficult to determine the CREF 
pivots of the kernel unless the initial guess is very accurate.

The structure of the kernel is an important consideration when 
deciding upon an initial guess. It is preferable to restrict fewer 
coefficients, however the iteration requires a better initialization due to the 
increased number of possible descent directions. In such scenarios for maximum 
flexibility, a globalized variant of Newton's method is required.  Like-wise, 
the  structure for $\Delta \A$ is also an important choice. 
Restricting which terms can be changed has a large influence on the 
(approximate) distance to singularity (of prescribed kernel dimension).

Another way to approach the lower-rank 
approximation problem is to use alternating projections or alternating 
directions of descent (since the objective is bi-linear with bi-linear 
constraints, it is convex in each argument) on the rank 
factorization in Section~\ref{sec:rankFact}. Since solutions in one coordinate, 
$\Delta \A$ are isolated, one can expect linear convergence with a reasonable 
algorithm. The lack-of normalization required overcomes the difficulty of 
choosing a suitable kernel displacement structure, however convergence would be 
linear at best and determining the dimensions of $U$ and $V$ is another problem 
to be discussed.   It is also worth noting that 
Algorithm~\ref{alg:newton-kernel} requires more computational resources per 
iteration as 
$r$ increases, however a rank factorization requires fewer computational 
resources per iteration as $r$ increases. 
 

\section{Conclusions and Future Work}

We have shown that finding lower-rank approximations of matrix polynomials can 
be established as a numerically well-posed problem and is amenable to first and 
second order optimization methods. The existence and isolation of solutions is 
established along with an algorithm exploiting affine structures to obtain 
locally quadratic convergence under mild normalization assumptions.  

Along with considering the lower-rank approximation of matrix polynomials, we 
present a generalization of the theory to matrix polynomials with an arbitrary 
affine structure. We provide examples of how the  
structure of permissible perturbations and prescribed kernel structure impacts 
the distance to solutions.

We also regard this current paper as a first step towards a formally
robust approach to non-linear matrix polynomials, in the spirit of
recent work with symbolic-numeric algorithms for polynomials.
Problems such as approximate matrix polynomial division, GCRD and
factorization all have applications which can benefit from these
modern tools.







\section*{References}
\bibliographystyle{elsarticle-harv.bst}

\end{document}